\documentclass{article}
\usepackage[utf8]{inputenc}
\usepackage{amsmath}
\usepackage{amsfonts}
\usepackage{amsthm}
\usepackage{enumitem}
\usepackage{bbm}
\usepackage[letterpaper,hmargin=1in,vmargin=1.25in]{geometry}
\usepackage{multirow}
\usepackage{comment}
\usepackage{biblatex}
\usepackage{authblk}

\newtheorem{theorem}{Theorem}
\newtheorem{lemma}[theorem]{Lemma}

\newtheorem{proposition}[theorem]{Proposition}
\newtheorem{corollary}[theorem]{Corollary}
\newtheorem{definition}[theorem]{Definition}

\title{A Probabilistic Proof of the nCPA to CCA Bound}
\author[1]{Ben Morris}
\author[2]{Hans Oberschelp}
\affil[1]{Department of Mathematics, University of California, Davis}
\affil[2]{Department of Mathematics, University of California, Davis}
\date{}

\begin{document}

\maketitle

\begin{abstract}
    We provide a new proof of Maurer, Renard, and Pietzak's result that the sum of the nCPA advantages of random permutations $P$ and $Q$ bound the CCA advantage of $P^{-1} \circ Q$. Our proof uses probability directly, as opposed to information theory, and has the advantage of providing an alternate sufficient condition of low CCA advantage. Namely, the CCA advantage of a random permutation can be bounded by its separation distance from the uniform distribution. We use this alternate condition to tighten the best known bound on the security of the swap-or-not shuffle in the special case of having fewer queries than the square root of the number of cards.
\end{abstract}

\section{Introduction}
The nCPA to CCA bound from Maurer, Pietzak, and Renner \cite{ind_amp} is a powerful tool for constructing secure cryptographic algorithms. It proves that an algorithm which is equivalent to its own inverse and which is secure against an nCPA adversary after $r$ rounds is also secure against a CCA adversary after $2r$ rounds. This means that nCPA and CCA security are essentially equivalent, up to increasing the runtime by a constant factor. In the first section of this paper, we provide an alternative proof of this bound using techniques from probability. \\
\\
Doubling the runtime to go from nCPA to CCA security does not always give a tight bound. In fact many well known random permutations have nearly identical nCPA and CCA security. If doubling the runtime is undesirable, a lemma in the first part of our paper provides an opportunity to get a bound on CCA security directly, without proving nCPA security and doubling the runtime. Specifically, a random permutation with low separation distance from uniform also has low CCA advantage. In the second part of our paper we utilize this lemma to improve the best known bound of CCA security of the Swap-or-Not shuffle when the number of queries is less than the square root of the number of cards. The Swap-or-Not shuffle was first described by Hoang, Morris, and Rogaway \cite{swaporiginal} as a private key encryption scheme that quickly encrypts medium-length strings with high, provable security.

\section{nCPA to CCA}
In this section we provide an alternative proof of the nCPA to CCA bound. We begin by defining nCPA and CCA security in the language of probability. In particular, we give a precise definition of CCA security, created to meld well with our subsequent proof. The proof is roughly split into two parts. First we prove a technical lemma about Markov chains, which, when translated into the language of cryptography, states that separation distance is small when two ciphers with good nCPA security are composed. Then, we show that CCA security is small when separation distance is small.

\subsection{Definitions}
We begin by introducing the notion of nCPA and CCA security with a game. \\
\\
Imagine you have two machines, Machine H and Machine T. Machine H generates a uniformly random permutation $U$ of the numbers $\{1,\dots,n\}$. You can query machine H by inputting any one of the numbers $\{1,\dots,n\}$. If you input $5$ then machine $A$ will output $U(5)$ i.e. the number that $U$ permutes $5$ to. Machine H only generates $U$ one time so if you input $5$ again you will get the same output, and if you input $7$ you will get a different output than the one from $5$. \\
\\
Machine T works exactly as Machine H except that it independently generates a random permutation $X$ according to some pre-established distribution of your choosing. \\
\\
You play a game against an opponent we will call the ``adversary'', or $A$ for short. At the start of the game you flip a fair coin. Then, the adversary provides you will a sequence of $q$ queries, which are numbers they want to input into one of the machines. If you flipped Heads at the start of the game, input their queries into Machine H and tell the adversary the results. If you flipped Tails, input their queries into Machine T and tell the adversary the results. Now the adversary guesses if you flipped Heads or Tails. We say ${\rm nCPA}_{q,A}(X)$ is the \textbf{non-adaptive chosen plaintext attack advantage} of $A$ against random permutation $X$ and we define this such that
\begin{equation*}
	2 \cdot {\rm nCPA}_{q,A}(X) - 1
\end{equation*}
is the probability of $A$ winning the game. Note that ${\rm nCPA}_{q,A}(X)$ is normalized so that ${\rm nCPA}_{q,A}(X) = 1$ if $A$ always wins the game and $\text{\rm nCPA}_{q,A}(X) = 0$ if $A$ utilizes the naive strategy of always guessing Heads. We say
\begin{equation*}
	{\rm nCPA}_q(X) = \max\limits_A\{\text{\rm nCPA}_{q,A}(X)\}
\end{equation*}
Note that $\text{\rm nCPA}_q(X)$ is close to $0$ if the distribution of $X$ is close to the uniform distribution. \\
\\
Now we define ${\rm CCA}_q(X)$ or the \textbf{chosen ciphertext attack advantage} against $X$. This is defined in exactly the same way as ${\rm nCPA})_q(X)$ except with two rule changes to the game:
\begin{itemize}
	\item The adversary can make some or all of their queries to the inverse permutation. Specifically, the adversary can provide a number $c \in \{1,\dots,n\}$ and specify that they want a ``reverse query'' and you must provide them with $U^{-1}(c)$ if you flipped Heads or $X^{-1}(c)$ if you flipped tails.
	\item The adversary is allowed to provide their queries one at a time, adapting their choice of next query based on the information they have received. For example, the adversary may first ask for a reverse query of the number $5$. When they are provided with the number $3$ as the response, they may use that to decide that they want to query the number $2$ in the normal forwards direction as their second query. This continues until they have exhausted all $q$ of their queries.
\end{itemize}
Note that the CCA advantage against $X$ must be higher than the nCPA advantage against $X$. This is because the adversary has strictly more tools at their disposal in the CCA version of the game, and so an optimal adversary will have a better chance at distinguishing $X$ from $U$. In fact, there are examples of distributions for $X$ where the nCPA advantage is close to 0 and the CCA advantage is close to 1. \\
\\
We now redefine both nCPA and CCA advantage in the language of probability. These definitions will be equivalent to the ones described above.
\begin{definition} If two finite random variables $X$ and $Y$ take in the same set $\mathcal{V}$, they have \textbf{total variation distance} given by
	\begin{equation*}
		d_{\rm{TV}}(X,Y) = \sup\limits_{A \subset \mathcal{V}} \Big(\mathbb{P}(X \in A) - \mathbb{P}(Y \in A)\Big).
	\end{equation*}
	We can equivalently define total variation distance by
	\begin{align*}
		d_{\rm{TV}}(X,Y) &= \frac{1}{2}\sum\limits_{a \in \mathcal{V}} \Big| \mathbb{P}(X = a) - \mathbb{P}(Y = a) \Big| \\
		&= \sum\limits_{a \in \mathcal{V}} \Big(\mathbb{P}(X = a) - \mathbb{P}(Y = a) \Big)^+.
	\end{align*}
\end{definition}
Note that total variation distance is a metric, and in particular $d_{\rm{TV}}(X,Y) = d_{\rm{TV}}(Y,X)$.
\begin{definition} If two finite random variables, X and Y, are defined taking values in the same set $\mathcal{V}$, the \textbf{separation distance} from $X$ to $Y$ is given by
	\begin{equation*}
		d_{\rm{sep}}(X,Y) = \sup\limits_{a \in \mathcal{V}} \left( 1 - \frac{\mathbb{P}(X = a)}{\mathbb{P}(Y = a)} \right)
	\end{equation*}
	where $\frac{x}{0} := 1$. \\
	\rm Note that separation distance is not a metric, as $d_{\rm{sep}}(X,Y)$ does not necessarily equal $d_{\rm{sep}}(Y,X)$.
\end{definition}
\begin{definition} For two finite sets, $S$ and $\mathcal{V}$, let $\mathcal{X} = \{ X(i) : i \in S  \}$ be a collection of random variables, all taking values in $\mathcal{V}$, and let $Y$ be another random variable taking values in $\mathcal{V}$. Then we define
	\begin{equation*}
		d_{\rm{TV}}(\mathcal{X},Y) = \max\limits_{i \in S} \ d_{\rm{TV}} (X(i),Y).
	\end{equation*}
	For separation distance we similarly define
	\begin{equation*}
		d_{\rm{sep}}(\mathcal{X},Y) = \max\limits_{i \in S} \ d_{\rm{sep}} (X(i),Y).
	\end{equation*}
\end{definition}
\begin{definition} Let $X$ be a random permutation of length $n$. Let $S_q$ be the set of all ordered $q$-tuples of $\{1,\dots ,n\}$. For $p = (p_1, \dots ,p_q) \in S_q$, Let $X(p)$ be the random vector $\Big( X(p_1), \dots ,X(p_q) \Big)$. Let $\mu$ be a uniform random element of $S_q$. Let $\mathcal{X}$ be the set of all $X(p)$ for $p \in S_q$. The \textbf{nCPA-security} of $X$ with $q$ queries is defined by
	\begin{equation*}
		\rm{nCPA}_q(X) = d_{\rm{TV}}(\mathcal{X},\mu).
	\end{equation*}
\end{definition}
Note $\text{nCPA}_n(X) = d_\text{TV}(X,U)$ where $U$ is the uniform random permutation. \\
\\
We will encode CCA queries to a permutation as a string of the form ``number, arrow, number''. For example, the notation $3 \rightarrow 5$ will be used if an adversary queries the image of $3$ and $\pi(3) = 5$. The notation $7 \leftarrow 2$ will be used if an adversary queries the preimage of $7$ and $\pi(2) = 7$.
\begin{definition} We will define $\mathcal{N}_n$ to be the space of CCA queries to a permutation of length $n$. Specifically, let $\mathcal{N}_n$ be the following set of 3-symbol strings,
	\begin{equation*}
		\mathcal{N}_n := \{aRb \ : \ a \in \{1,\dots ,n\}, R \in \{\rightarrow,\leftarrow\}, b \in \{1,\dots ,n\} \}.
	\end{equation*}
	We call the first two symbols of $p \in \mathcal{N}_n$ the \textbf{input}, which we denote $I(p)$. For example, $I(3 \rightarrow 5) = 3 \rightarrow$. We call the last entry the \textbf{output}, which we denote $O(p)$. For example, $O(7 \leftarrow 2) = 2$. \\
	\\
	For $a,b \in \{1, \dots, n\}$ we say $a \rightarrow b$ and $b \rightarrow a$ are \textbf{reversals} of each other. We say two CCA queries $p_1$ and $p_2$ are \textbf{equivalent} (and we write $p_1 \sim p_2$) if $p_1 = p_2$ or $p_1$ and $p_2$ are reversals of each other. \\
	\\
	\rm Note that $\sim$ gives an equivalence relation on $\mathcal{N}_n$.
\end{definition}
\begin{definition} A function $f:S_n \rightarrow {\mathcal{N}_n}^q$ is called a $q$-query CCA \textbf{strategy} if for every $k \in \{1,\dots ,n\}$ and $\sigma, \tau \in S_n$ the following statements hold:
	\begin{enumerate} \rm
		\item if $f(\sigma)_k \sim (a \rightarrow b)$ then $\sigma(a) = b$,
		\item $I(f(\cdot)_1)$ is constant, i.e. $I(f(\sigma)_1)$ does not depend on $\sigma$;
		\item if $(f(\sigma)_1,\dots ,f(\sigma)_{k-1})$ = $(f(\tau)_1,\dots ,f(\tau)_{k-1})$, then $I(f(\sigma)_{k})=I(f(\tau)_{k})$.
	\end{enumerate}
\end{definition}
A strategy is a way an adversary might make $q$ queries to an unknown permutation. At first the adversary knows nothing, so the question of the first query does not depend on the permutation. The first question the adversary asks is ``where does this permutation (or, if the adversary so chooses, the inverse of this permutation) send the element $a$?'' The result of the first query tells the adversary the answer to this question. Then, the adversary's second question can be based on the information gained by the first query. The adversary's third question can be based on the information gained by the first two queries, and so forth.
\begin{definition} Let $X$ be a random permutation of length $n$. The \textbf{CCA-security} of X with $q$ queries is given by
	\begin{equation*}
		\text{CCA}_q(X) = \max\limits_{f \text{ \rm is a $q$-query strategy}} d_{\rm{TV}}(f(X),f(U)),
	\end{equation*}
	where U is the uniform random permutation of length $n$.
\end{definition}

\subsection{Technical Lemmas}

In this subsection we prove some technical lemmas regarding Markov chains and random permutations. In the first two results we show an upper bound for separation distance of the composition of two Markov chains in terms of the total variation distances of the individual chains. \\
\\
\begin{lemma} Let $P,Q$ be Markov chains on state space $S$ where $S$ is finite, and suppose $P,Q$ both have stationary distribution $\pi$. Let $\overleftarrow{P}$ be the time reversal of $P$. Then for all $i,j \in S$,
	\begin{equation*}
		1 - \frac{Q\overleftarrow{P}(i,j)}{\pi(j)} \leq d_{\rm{TV}}(P(j,\cdot),\pi) +d_{\rm{TV}}(Q(i,\cdot), \pi).
	\end{equation*}
\end{lemma}
\begin{proof}
	Fix any $i,j \in S$. Then,
	\begin{align}
		Q\overleftarrow{P}(i,j) &= \sum\limits_{z \in S} Q(i,z) \cdot \overleftarrow{P}(z,j) \label{conditiononz} \\
		&= \sum\limits_{z \in S} Q(i,z) \cdot \frac{\pi(j)}{\pi(z)}\cdot P(j,z) \label{timereversal} \\
		&= \pi(j) \sum\limits_{z \in S} \frac{ P(j,z)}{\pi(z)}\cdot \frac{Q(i,z)}{\pi(z)} \cdot \pi(z),
	\end{align}
	where (\ref{conditiononz}) comes from conditioning on the state after the Q step, and (\ref{timereversal}) uses the definition of the time reversal. Let
	\begin{align*}
		\Delta_P(z) &:= \frac{P(j,z) - \pi(z)}{\pi(z)}, \\
		\Delta_Q(z) &:= \frac{Q(i,z) - \pi(z)}{\pi(z)}.
	\end{align*}
	Then
	\begin{align*}
		\frac{ P(j,z)}{\pi(z)} &= 1+\Delta_P(z),\\
		\frac{ Q(i,z)}{\pi(z)} &= 1+\Delta_Q(z),
	\end{align*}
	and hence
	\begin{align}
		\frac{Q\overleftarrow{P}(i,j)}{\pi(j)} &= \sum\limits_{z \in S} (1 + \Delta_P(z)) \cdot (1 + \Delta_Q(z)) \cdot \pi(z) \\
		&= \sum\limits_{z \in S} \pi(z) + \sum\limits_{z \in S} \Delta_P(z)\pi(z) + \sum\limits_{z \in S} \Delta_Q(z)\pi(z) + \sum\limits_{z \in S} \Delta_P(z)\Delta_Q(z)\pi(z). \label{lasttermqpsum}
	\end{align}
	Since $\pi(z)$ is a probability vector we have $\sum_{z \in S} \pi(z) = 1$. Furthermore, $P(j,\cdot)$ is also a probability vector so
	$$
	\sum_{z \in S} \Delta_P(z)\pi(z) = \sum_{z \in S} (P(j,z) - \pi(z)) = 0.
	$$
	Similarly, $\sum_{z \in S} \Delta_Q(z)\pi(z) =  0$. To bound the final sum in (\ref{lasttermqpsum}) note that
	\begin{align}
		\sum\limits_{z \in S} \Delta_P(z)\Delta_Q(z)\pi(z) \geq \sum\limits_{z \in S} -\Big( \Delta_P(z)\Delta_Q(z)\pi(z) \Big)^-. \label{qpsumneg}
	\end{align}
	For every nonzero term on the right hand side of (\ref{qpsumneg}), either $\Delta_P(z) > 0$ and $\Delta_Q(z) < 0$, or $\Delta_Q(z) > 0$ and $\Delta_P(z) < 0$. This gives us
	\begin{align}
		\sum\limits_{z \in S} -\Big( \Delta_P(z)\Delta_Q(z)\pi(z) \Big)^- &= \sum\limits_{\substack{\Delta_P(z) > 0 \\ \Delta_Q(z) < 0}} \Delta_P(z)\Delta_Q(z)\pi(z) \ + \sum\limits_{\substack{\Delta_Q(z) > 0 \\\Delta_P(z) < 0}} \Delta_P(z)\Delta_Q(z)\pi(z) \\
		&\geq \sum\limits_{\substack{\Delta_P(z) > 0 \\ \Delta_Q(z) < 0}} -\Delta_P(z)\pi(z) \ + \sum\limits_{\substack{\Delta_Q(z) > 0 \\\Delta_P(z) < 0}} -\Delta_Q(z)\pi(z), \label{splitsums}
	\end{align}
	where (\ref{splitsums}) comes from the fact that $\Delta_Q(z) \geq -1$ and $\Delta_P(z) \geq -1$ for all $z$. Finally, note that
	\begin{align}
		\sum\limits_{\substack{\Delta_P(z) > 0 \\ \Delta_Q(z) < 0}} -\Delta_P(z)\pi(z) &\geq \sum\limits_{\Delta_P(z) > 0} -\Delta_P(z)\pi(z) \\
		&= \sum\limits_{P(j,z) > \pi(z)} -(P(j,z) - \pi(z)) \\
		&= -d_{\rm{TV}}(P(j, \cdot ),\pi). \label{pbound}
	\end{align}
	A similar argument shows that
	\begin{equation}
		\sum\limits_{\substack{\Delta_Q(z) > 0 \\\Delta_P(z) < 0}} -\Delta_Q(z)\pi(z) \geq -d_{\rm{TV}}(Q(i,\cdot), \pi) \label{qbound}
	\end{equation}
	Combining (\ref{pbound}) and (\ref{qbound}) with (\ref{qpsumneg}) and (\ref{splitsums}) gives,
	\begin{equation}
		\frac{Q\overleftarrow{P}(i,j)}{\pi(j)} \geq 1 -d_{\rm{TV}}(P(j,\cdot ),\pi) -d_{\rm{TV}}(Q(i,\cdot ), \pi)
	\end{equation}
	and the lemma follows.
\end{proof}

\begin{corollary}\label{sepvstv} Let $P,Q$ be Markov chains on a finite state space $S$, both with the stationary distribution $\pi$. Let $\overleftarrow{P}$ be the time reversal of $P$. Let $\mathcal{P} := \{ P(i,\cdot) \}_{i \in S}$, $\mathcal{Q} := \{ Q(i,\cdot) \}_{i \in S}$, and $\mathcal{Q}\overleftarrow{\mathcal{P}} := \{ Q\overleftarrow{P}(i,\cdot ) \}_{i \in S}$. Then,
	\begin{enumerate} \rm
		\item for all $i \in S$ we have $d_{\rm sep}(Q\overleftarrow{P}(i,\cdot),\pi) \leq d_{\rm TV}(\mathcal{P},\pi) + d_{\rm TV}(Q(i),\pi)$
		\item $d_{\rm sep}(\mathcal{Q}\overleftarrow{\mathcal{P}},\pi) \leq d_{\rm TV}(\mathcal{P},\pi) + d_{\rm TV}(\mathcal{Q},\pi)$.
	\end{enumerate}
\end{corollary}
\begin{proof} By Lemma 8, for all $i,j \in S$,
	\begin{equation}
		1 - \frac{Q\overleftarrow{P}(i,j)}{\pi(j)} \leq d_{\rm{TV}}(P(j,\cdot),\pi) +d_{\rm{TV}}(Q(i,\cdot), \pi).
	\end{equation}
	Taking the maximum of both sides over $j$ gives
	\begin{align}
		\max\limits_{j \in S} \left[ 1 - \frac{Q\overleftarrow{P}(i,j)}{\pi(j)} \right] &\leq \max\limits_{j \in S} \Big[  d_{\rm{TV}}(P(j,\cdot),\pi) +d_{\rm{TV}}(Q(i,\cdot), \pi) \Big], \\
		d_{\rm{sep}}(Q\overleftarrow{P}(i,\cdot),\pi) &\leq d_{\rm TV}(\mathcal{P},\pi) + d_{\rm{TV}}(Q(i,\cdot),\pi).
	\end{align}
	This is our first result. Now we take the maximum of both sides over $i$ and have
	\begin{align}
		\max\limits_{i \in S} \  d_{\rm{sep}}(Q\overleftarrow{P}(i,\cdot),\pi) &\leq \max\limits_{i \in S} \Big[ d_{\rm TV}(\mathcal{P},\pi) + d_{\rm{TV}}(Q(i,\cdot),\pi) \Big], \\
		d_{\rm sep}(\mathcal{Q}\overleftarrow{\mathcal{P}},\pi) &\leq d_{\rm TV}(\mathcal{P},\pi) + d_{\rm TV}(\mathcal{Q},\pi).
	\end{align}
\end{proof}

The next three results show that CCA advantage is bounded above by separation distance. We will later combine this fact with the prior results from this section to achieve a bound on CCA security in terms of nCPA security. It is also a useful fact in its own right because it gives us a tight bound on CCA security using a well-studied metric from probability.

\begin{lemma}\label{strategyconsistent} Let $\sigma$ be a permutation of length $n$. Let $f$ be a CCA strategy with $q$ queries. Let $p = (p_1, \dots ,p_q) \sim (a_1 \rightarrow b_1, \dots ,a_q \rightarrow b_q) \in \mathcal{N}^q$, and suppose that $p$ is in the image of $f$. Then,
	\begin{align*}
		f(\sigma) = p \text{ if and only if } \sigma(a_1, \dots ,a_q) = (b_1, \dots ,b_q)
	\end{align*}
\end{lemma}
\begin{proof} First we assume $f(\sigma) = p$. Then $f(\sigma) \sim (a_1 \rightarrow b_1, \dots, a_q \rightarrow b_q)$. The definition of a strategy requires $\sigma(a_i) = b_i$ for all $i$. So if $f(\sigma) = p$ then $\sigma(a_1, \dots ,a_q) = (b_1, \dots ,b_q)$. \\
	\\
	Now we assume $f(\sigma) = \ell = (\ell_1,\dots,\ell_q) \neq p = (p_1,\dots,p_q)$. Let $m = \min\{i : \ell_i \neq p_i \}$. Note that for all $j < m$ we have $\ell_j = p_j$. This, along with $p$ and $\ell$ being in the image of the same strategy, means $I(\ell_m) = I(p_m)$. This implies $O(\ell_m) \neq O(p_m)$. We now consider two cases:
	\begin{itemize}
		\item \textbf{Case 1} \\
		If $p_m = a_m \rightarrow b_m$, then $\ell_m = a_m \rightarrow c_m$ where $c_m \neq b_m$. So $\sigma(a_m) = c_m \neq b_m$.
		\item \textbf{Case 2} \\
		If $p_m = b_m \leftarrow a_m$, then $\ell_m = b_m \leftarrow d_m$ where $d_m \neq a_m$. So $\sigma(a_m) \neq \sigma(d_m) = b_m$ because $\sigma$ is a permutation.
	\end{itemize}
	Either way $\sigma(a_m) \neq b_m$, hence if $f(\sigma) \neq p$ then $\sigma(a_1, \dots ,a_q) \neq (b_1, \dots ,b_q)$.
\end{proof}

\begin{corollary}\label{strategyconsistentcor}
	Let $f$ be a CCA strategy with $q$ queries. Let $\Phi \subset{\mathcal{N}_n}^q$ be the image of $f$. Let $S$ be the set of all ordered $q$-tuples of distinct elements of $\{1, \dots, n\}$. Then there is a one-to-one correspondence between $\Phi$ and a subset $H_f \subset S^2$ where each $p \in \Phi$ is matched with $(a,b) \in H_f$ such that
	\begin{equation}
		f(\sigma) = p \text{ if and only if } \sigma(s_1) = s_2.
	\end{equation}
\end{corollary}

\begin{proof}
	By Lemma \ref{strategyconsistent} we already know that for each $p \in \Phi$ there exists $(a,b) \in S^2$ such that
	\begin{equation}
		f(\sigma) = p \text{ if and only if } \sigma(s_1) = s_2.
	\end{equation}
	All that remains is to show that this mapping is injective. Suppose $p,p' \in \Phi$ such that $p \neq p'$. Let $k$ be the minimal value of $\{1,\dots,q\}$ such that $p_k \neq p_k'$. Since $p,p'$ are both in the image of $f$ we have that $I(p_1) = I(p_1')$. In addition, if $k \geq 2$ then we know that $(p_1,\dots,p_{k-1}) = (p_1',\dots,p_{k-1}')$ and since $p,p'$ are in the image of $f$ we have $I(p_k) = I(p_k')$. Without loss of generality assume that $I(p_k),I(p_k')$ both take the form $a \rightarrow$ for some $a \in \{1,\dots,n\}$. Since $p_k \neq p_k'$ there must exist $b,b' \in \{1,\dots,b\}$ such that $b \neq b'$ and
	\begin{equation}
		p_k = a \rightarrow b \text{ and } p_k' = a \rightarrow b'.
	\end{equation}
	So the following statements hold:
	\begin{itemize}
		\item if $f(\sigma) = p$ then $\sigma(a) = b$,
		\item if $f(\sigma) = p'$ then $\sigma(a) = b'$.
	\end{itemize}
	Therefore $s,s' \in S^2$ associated with $p,p'$ respectively cannot be the same.
\end{proof}

\begin{theorem}\label{ccasepineq}
	Let $X$ be a random permutation of length $n$. Let $S$ be the set of all ordered $q$-tuples of $\{1,\dots ,n\}$. Let $\mathcal{X} := \{ X(p) \}_{p \in S}$. Let $\mu_q$ be the uniform distribution on $S$. Then,
	\begin{equation*}
		\text{CCA}_q(X) \leq d_{\rm sep}(\mathcal{X},\mu).
	\end{equation*}{}
\end{theorem}
\begin{proof}
	Fix some $q$-query strategy $f$. Assume $f$ is optimal (total variation distance maximizing). Let $\Phi \subset{\mathcal{N}_n}^q$ be the image of $f$. By the optimality of $f$ we can assume that there does not exist $p \in \Phi$ such that $p_i \sim p_j$ for any $i \neq j$. (This is because no optimal strategy would ever ask a question it already knows the answer to. In other words, if $p_i \sim (3 \rightarrow 5)$, then no optimal strategy would ask $(3 \rightarrow)$ or $(5 \leftarrow)$ as $I(p_j)$ for $j > i$.) \\
	\\
	First we compute $|\Phi|$. We can count all $p \in \Phi$ as follows: There is,
	\begin{align*}
		1 \hspace{8mm} &\text{ possible value of } I(p_1), \\
		n \hspace{8mm} &\text{ possible values of } O(p_1), \\
		1 \hspace{8mm} &\text{ possible value of } I(p_2), \text{ given } p_1 \\
		(n-1) \hspace{3mm} &\text{ possible values of } O(p_2) \text{ given } p_1 \text{ and } I(p_2), \\
		&\hspace{2cm} \vdots \\
		1 \hspace{8mm} &\text{ possible value of } I(p_q) \text{ given } p_1, \dots, p_{q-1}, \\
		(n-q+1) \hspace{1mm} &\text{ possible values of } O(p_q) \text{ given } p_1, \dots, p_{q-1} \text{ and } I(p_q).
	\end{align*}
	So $|\Phi| = n(n-1) \dots (n-q+1) =: (n)_q$. We will set this result aside for now. \\
	\\
	Using the definition of total variation distance,
	\begin{align}
		d_{\rm{TV}}(f(X),f(U)) &= \sum_{p \in \Phi} \Big[\mathbb{P}(f(U) = p) - \mathbb{P}(f(X) = p)\Big]^+.
	\end{align}
	Lemma \ref{strategyconsistent} tells us that for each $p \in \Phi$ there exists $(a,b) = ((a_1,\dots,a_q),(b_1,\dots,b_q)) \in S^2$ such that
	\begin{equation*}
		\Big[ \mathbb{P}(f(U) = p) - \mathbb{P}(f(X) = p) \Big] = \Big[\mathbb{P}(U(a)=b) - \mathbb{P}(X(a) = b)\Big].
	\end{equation*}
	Let $H_f$ be the set of all such $(a,b)$. Then by Corollary \ref{strategyconsistentcor} we have $|H_f| = |\Phi| = (n)_q$ and
	\begin{align}
		d_{\rm{TV}}(f(X),f(U)) =& \sum_{(a,b) \in H_f} \Big[\mathbb{P}(U(a) = b) - \mathbb{P}(X(a) = b) \Big]^+ \\
		&= \sum_{(a,b) \in H_f} \left[\frac{1}{(n)_q} - \mathbb{P}(X(a) = b) \right]^+ \\
		&= \frac{1}{(n)_q} \sum_{(a,b) \in H_f} \left[1 - \frac{\mathbb{P}(X(a) = b)}{{(n)_q}^{-1}} \right]^+.
	\end{align}
	If we replace each term in the sum with the maximum over all $(a,b) \in H_f$, we get the inequality
	\begin{align}
		d_{\rm{TV}}(f(X),f(U)) &\leq \frac{1}{(n)_q} |H_f| \max\limits_{(a,b) \in H_f}\left|1 - \frac{\mathbb{P}(X(a) = b)}{{(n)_q}^{-1}} \right| \\
		&= \max\limits_{(a,b) \in H_f}\left|1 - \frac{\mathbb{P}(X(a) = b)}{{(n)_q}^{-1}} \right| \\
		&\leq \max\limits_{(a,b) \in S^2}\left|1 - \frac{\mathbb{P}(X(a) = b)}{{(n)_q}^{-1}} \right| .\label{sepbound}
	\end{align}
	Using the definition of separation distance we can rewrite $(\ref{sepbound})$ as
	\begin{align}
		d_{\rm{TV}}(f(X),f(U)) &\leq \max\limits_{a \in S} d_{\rm{sep}}( X(a), U(a) ) \\
		&= d_{\rm sep}(\mathcal{X},\mu)
	\end{align}
	Since this inequality holds for all strategies $f$, we get
	\begin{equation}
		\text{CCA}_q(X) \leq d_{\rm sep}(\mathcal{X},\mu)
	\end{equation}
\end{proof}

\subsection{Main Theorem}
We now have all the tools necessary to prove the nCPA to CCA bound, the main result of this section.
\begin{theorem} Let $X,Y$ be random permutations of length $n$. Let $q \in \{1,\dots ,n\}$. Then
	\begin{equation*}
		\text{\rm CCA}_q(X^{-1} \circ Y) \leq \text{\rm nCPA}_q(X) + \text{\rm nCPA}_q(Y) \label{ccavssep}
	\end{equation*}
\end{theorem}
\begin{proof} This is a straightforward application of Corollary \ref{sepvstv} and Theorem \ref{ccasepineq}. Let $S$ be the set of all ordered $q$-tuples of distinct elements of $\{1,\dots ,n\}$. Let $\mathcal{X}:= \{ X(p) \}_{p \in S}$ and $\mathcal{Y} := \{ Y(p) \}_{p \in S}$ and $\mathcal{X}^{-1}\mathcal{Y} := \{ X^{-1} \circ Y(p) \}_{p \in S}$. Let $\mu_q$ be the uniform distribution on $S$. Then from Theorem \ref{sepvstv} we have
	\begin{align}
		\text{CCA}_q(X^{-1} \circ Y) \leq d_{\rm sep}(\mathcal{X}^{-1}\mathcal{Y},\mu_q). \label{ccaseppart}
	\end{align}
	We can think of $X$ and $Y$ each as one step of a Markov Chain on $S_n$. Then by Corollary \ref{sepvstv} we have,
	\begin{equation}
		d_{\rm sep}(\mathcal{X}^{-1}\mathcal{Y},\mu) \leq d_{\rm TV}(\mathcal{X},\mu_q) + d_{\rm TV}(\mathcal{Y},\mu_q). \label{xyseptvbound}
	\end{equation}
	By applying the definition of $\text{nCPA}_q$ the right hand side of (\ref{xyseptvbound}) we have
	\begin{equation}
		d_{\rm sep}(\mathcal{X}^{-1}\mathcal{Y},\mu) \leq \text{\rm nCPA}_q(X) + \text{\rm nCPA}_q(Y). \label{sepncpapart}
	\end{equation}
	Combining (\ref{ccaseppart}) and (\ref{sepncpapart}) completes the theorem.
\end{proof}

\section{Bound on separation distance of the Swap-or-Not Shuffle}
Hoang, Morris, and Rogaway \cite{swaporiginal} proved that, in a message space of size $N$, the swap-or-not shuffle can achieve strong CCA security after approximately $r = 6\log_2(N)$ rounds for $q<N^{1-\epsilon}$. In 2017, Dai, Hoang, and Tessaro \cite{swapimproved} improved the bound, and showed that only $r = 4\log_2(N)$ rounds are required. In this section, we show that approximately $r = \log_2(N)$ rounds is sufficient provided that the number of queries is less than $\sqrt{N}$. This upper bound on the number of rounds required for strong security is tight when the number of queries is more than $\log_2(N)$.

\subsection{Definition of the Swap on Not Shuffle}
The swap-or-not shuffle is a random permutation defined as follows: We start with a deck of $N=2^d$ cards, and a collection of vectors $\mathcal{K}_1,\dots,\mathcal{K}_r \in \mathbb{Z}_2^d$ which we call \textit{round keys}. First, label each card as a unique element of $\mathbb{Z}_2^d$. The specific labeling is not important to the security of the shuffle, so for simplicity we will label cards by their initial position in the deck, in binary. So in a shuffle of $16$ cards, card $0010$ is initially in position $0010$ (or the $3$rd topmost card). In round $j$, let the cards in positions $x$ and $y$ be ``paired'' with respect to round key $\mathcal{K}_j$ if $x + y =\mathcal{K}_j$ (where addition is done in $\mathbb{Z}_2^d$). Then for each pair flip an independent coin, and if Heads swap the positions of the cards in the pair, and if Tails do nothing. Repeat this for $r$ independent rounds. \\
\\Denote $x^t(\mathcal{K}_1,\dots,\mathcal{K}_t)$ as the random element of $\mathbb{Z}_2^d$ which is the position of card $x$ (i.e. the card initially in position $x$) after $t$ steps of the shuffle using round keys $\mathcal{K}_1,\dots,\mathcal{K}_t$. Let $x^t := x^t(K_1,\dots,K_t)$ where $K_1,\dots,K_t$ are iid uniformly sampled from $\mathbb{Z}_2^d$. So $x^0 = x$. For $y \in \mathbb{Z}_2^d$ we write $x \rightarrow y$ for the event $x^r = y$.

\subsection{Round Keys are Likely to Span $\mathbb{Z}_2^d$}
We will fix a set of $q$ cards, with initial positions $x_1,x_2,\dots,x_q$. So $(x_1^r,x_2^r,\dots,x_q^r)$ is the random vector of positions of these $q$ cards after $r$ rounds. Define coins $c_{i,j}$ as follows:
\begin{equation}
	c_{i,j} = \begin{cases}
		1 \text{ if card } i \text{ if swapped in round } j \\
		0 \text{ otherwise}
	\end{cases}
\end{equation}
Then for round keys $\mathcal{K}_1,\dots,\mathcal{K}_r$ we have
\begin{align*}
	x_1^r &= x_1 + c_{1,1}\mathcal{K}_1 + c_{1,2}\mathcal{K}_2 + \dots + c_{1,r}K_r \\
	x_2^r &= x_2 + c_{2,1}\mathcal{K}_1 + c_{2,2}\mathcal{K}_2 + \dots + c_{2,r}\mathcal{K}_r \\
	& \hspace{2cm} \vdots \\
	x_q^r &= x_q + c_{q,1}\mathcal{K}_1 + c_{q,2}\mathcal{K}_2 + \dots + c_{q,r}\mathcal{K}_r
\end{align*}
Note that the coins $c_{i,j}$ are not independent. In particular, if $x_i^{t-1} + \mathcal{K}_t = x_j^{t-1}$ then $c_{i,t} = c_{j,t}$. We can see that the round keys $\mathcal{K}_1,\dots,\mathcal{K}_r$ need to span $\mathbb{Z}_2^d$ to make $x_1^r,\dots,x_q^r$ close to uniform. Otherwise each $x_i^r - x_i$ will be in the same subspace of $\mathbb{Z}_2^d$, which would be very unlikely for a uniform random permutation. Fortunately, it is very likely $K_1,\dots,K_r$ span $\mathbb{Z}_2^d$ as long as $r$ is slightly larger than $d$. We will now make this precise.
\begin{lemma}\label{Arbound}
	Fix $r \geq d$. Let $A_r$ be the event that $K_1,\dots,K_r$ span $\mathbb{Z}_2^d$. Then
	$$
	\mathbb{P}({A_r}) \geq 1-2^{d-r}
	$$
\end{lemma} 
\begin{proof}
	For any $v \in \mathbb{Z}_2^d$ let $H_v$ be the event that $v$ is orthogonal to each of $K_1,\dots,K_r$. Then,
	\begin{equation}
		A_r^C = \bigcup\limits_{v \neq 0} H_v.
	\end{equation}
	So,
	\begin{equation}
		\mathbb{P}(A_r^C) = \mathbb{P}\left( \bigcup\limits_{v \neq 0} H_v \right) \leq \sum\limits_{v \neq 0} \mathbb{P}(H_v).
	\end{equation}
	For each $v \neq 0$, we have $\mathbb{P}(H_v) = 2^{-r}$, as each $K_i$ is independently in or out of the plane $v^\perp$ with probability $\frac{1}{2}$ each. Since there are $2^d-1$ different vectors in the sum,
	\begin{equation}
		\mathbb{P}(A_r^C) \leq \sum\limits_{v \in \mathbb{Z}_2^d} 2^{-r} \leq 2^{d-r}
	\end{equation}
\end{proof}

The fact that the round keys are likely to span $\mathbb{Z}_2^d$ after $r$ rounds when $r$ is larger than $d$ should give us hope that the swap-or-not shuffle will be well-mixed after $r$ rounds. Indeed, we would know that the swap-or-not shuffle was perfectly mixed if only the coins $c_{i,j}$ were all independent. With this idea in mind, our strategy to prove the swap-or-not shuffle is well-mixed will proceed as follows:
\begin{itemize}
	\item First we define a new process, which is similar to swap-or-not shuffle but has independent coins.
	\item Then we show that this new process is uniform as long as the round keys span $\mathbb{Z}_2^d$\
	\item Finally we couple the swap-or-not shuffle to this new process in such a way that it is likely to stay coupled for all $r$ rounds.
\end{itemize}

\subsection{Collisions and the Tilde Process}
We now consider a variation on the swap-or-not shuffle, which is not strictly speaking a shuffle (that is, it is not a random permutation). Start with a deck of $n = 2^d$ cards, labeled by their initial positions in the deck. As before, if round keys $\mathcal{K}_1,\dots,\mathcal{K}_r \in \mathbb{Z}_2^d$, then let
\begin{align*}
	\widetilde{x_1^r}(\mathcal{K}_1,\dots,\mathcal{K}_t) &= x_1 + \widetilde{c_{1,1}}\mathcal{K}_1 + \widetilde{c_{1,2}}\mathcal{K}_2 + \dots + \widetilde{c_{1,r}}\mathcal{K}_r \\
	\widetilde{x_2^r}(\mathcal{K}_1,\dots,\mathcal{K}_t) &= x_2 + \widetilde{c_{2,1}}\mathcal{K}_1 + \widetilde{c_{2,2}}\mathcal{K}_2 + \dots + \widetilde{c_{2,r}}\mathcal{K}_r \\
	& \hspace{2cm} \vdots \\
	\widetilde{x_q^r}(\mathcal{K}_1,\dots,\mathcal{K}_t) &= x_q + \widetilde{c_{q,1}}\mathcal{K}_1 + \widetilde{c_{q,2}}\mathcal{K}_2 + \dots + \widetilde{c_{q,r}}\mathcal{K}_r
\end{align*}
where $\widetilde{c_{i,j}}$ are iid Bernoulli$(\frac{1}{2})$ random variables. In other words, if $x$ and $x + \mathcal{K}_j$ are paired, then instead of swapping places or remaining put with probability $\frac{1}{2}$ each, now $x$ and $x + \mathcal{K}_j$ will \textit{both} go to $x$, or \textit{both} go to $x + \mathcal{K}_j$, or swap, or stay put, with probability $\frac{1}{4}$ each. We call this process the tilde process (and we keep in mind it is not a random permutation because it is not necessarily injective). As before, we write $\widetilde{x^t}(\mathcal{K}_1,\dots,\mathcal{K}_t)$ as the (random) position of card $x$ under the tilde process after $t$ steps using round keys $\mathcal{K}_1,\dots,\mathcal{K}_t$. Let $\widetilde{x^t}$ be defined similarly but with iid uniform round keys. We write $x \widetilde{\rightarrow} y$ for the event $\widetilde{x^r} = y$.

\begin{lemma}\label{tildeuniform}
	Fix any $x_1,\dots,x_q,y_1,\dots,y_q \in \mathbb{Z}_2^d$. Also fix any $\mathcal{K}_1,\dots,\mathcal{K}_r \in \mathbb{Z}_2^d$ with $r \geq d$ such that $\mathcal{K}_1,\dots,\mathcal{K}_r$ span $\mathbb{Z}_2^d$. Consider the tilde process on $\mathbb{Z}_2^d$ with $r$ rounds. Let $K_1,\dots,K_d$ be the iid uniform round keys, and let $\widetilde{c_{i,j}}$ be the coins. Then,
	\begin{enumerate}
		\item For all $t$ the distribution of $\left(\widetilde{x_1^t}(\mathcal{K}_1,\dots,\mathcal{K}_t)+x_1,\dots,\widetilde{x_q^t}(\mathcal{K}_1,\dots,\mathcal{K}_t)+x_q\right)$ is uniform over $\left(\text{span}(\mathcal{K}_1,\dots,\mathcal{K}_t)\right)^q$.
		\item $\mathbb{P}(x_1 \widetilde{\rightarrow} y_1, \dots x_q \widetilde{\rightarrow} y_q \ | \ K_1 = \mathcal{K}_1, \dots, K_r = \mathcal{K}_r) = 2^{-qd}$
		\item $\mathbb{P}(\widetilde{c_{i,j}} = \mathcal{C}_{i,j} \text{ for all } i,j \ | \ x_1 \widetilde{\rightarrow} y_1, \dots x_q \widetilde{\rightarrow} y_q, \ K_1 = \mathcal{K}_1, \dots, K_r = \mathcal{K}_r) = 2^{q(d-r)}$ for all $(\mathcal{C}_{i,j})$ where $x_i + \mathcal{C}_{i,1}\mathcal{K}_1 + \dots + \mathcal{C}_{i,r}\mathcal{K}_r = y_i$ for all $i$. In other words if the round keys span $\mathbb{Z}_2^d$ then the coins are uniformly distributed across all ``valid'' choices that take each $x_i$ to $y_i$.
	\end{enumerate}
\end{lemma}
\begin{proof}
	We prove $(1)$ induction. For the base case, note that by definition, each $x_i^1 + x_i = \mathcal{K}_1^{c_{1,i}}$. Since $c_{1,1},\dots,c_{1,q}$ are independent, each $x_i^1 + x_i$ is independently equally likey to equal $0$ or $\mathcal{K}_1$. \\
	\\
	For the inductive step, assume $\left(\widetilde{x_1^t}(\mathcal{K}_1,\dots,\mathcal{K}_t)+x_1,\dots,\widetilde{x_q^t}(\mathcal{K}_1,\dots,\mathcal{K}_t)+x_q\right)$ is distributed uniformly across span$(\mathcal{K}_1,\dots,\mathcal{K}_t)$. In the case that $\mathcal{K}_{t+1} \in \text{span}(\mathcal{K}_1,\dots,\mathcal{K}_t)$, adding $c_{i,t+1}\mathcal{K}_{t+1}$ to each $\widetilde{x_i^t}$ amounts to adding a vector in the subspace span$(\mathcal{K}_1,\dots,\mathcal{K}_t)^q$ to a uniform random element of that subspace, and so the distribution will remain uniform. In the case that that $\mathcal{K}_{t+1} \not\in \text{span}(\mathcal{K}_1,\dots,\mathcal{K}_t)$, the $\mathcal{K}_{t+1}$ component of each $\widetilde{x_i^{t+1}}$ will equally likely be present or absent independently and the component orthogonal to $\mathcal{K}_{t+1}$ will remain uniform, so the distribution of $\left(\widetilde{x_1^{t+1}}(\mathcal{K}_1,\dots,\mathcal{K}_{t+1})+x_1,\dots,\widetilde{x_q^{t+1}}(\mathcal{K}_1,\dots,\mathcal{K}_{t+1})+x_q\right)$ will be uniform over span$(\mathcal{K}_1,\dots,\mathcal{K}_t)$ \\
	\\
	Now (2) follows immediately from (1) after setting $t = r$ and recalling that by assumption $|\text{span}(\mathcal{K}_1,\dots,\mathcal{K}_r)| = |\mathbb{Z}_2^d| = 2^{-d}$. \\
	\\
	To show (3), fix any $\mathcal{C}_{i,j} \in \{0,1\}$ such that $\widetilde{c_{i,j}} = \mathcal{C}_{i,j}$ and $K_j = \mathcal{K}_j$ for all $i,j$ imply $x_i \widetilde{\rightarrow} y_i$ for all $i$. Then,
	\begin{align}
		&\mathbb{P}(\widetilde{c_{i,j}} = \mathcal{C}_{i,j} \text{ for all } i,j \ | \ x_1 \widetilde{\rightarrow} y_1, \dots x_q \widetilde{\rightarrow} y_q, \ K_1 = \mathcal{K}_1, \dots, K_r = \mathcal{K}_r) \\
		= \ &\frac{\mathbb{P}(\widetilde{c_{i,j}} = \mathcal{C}_{i,j} \text{ for all } i,j \ | \ K_1 = \mathcal{K}_1,\dots,K_r = \mathcal{K}_r)}{\mathbb{P}(x_1 \widetilde{\rightarrow} y_1, x_2 \widetilde{\rightarrow} y_2, \dots, x_q \widetilde{\rightarrow} y_q \ | \ K_1 = \mathcal{K}_1,\dots,K_r = \mathcal{K}_r)} \label{removex2y} \\
		= \ &\frac{2^{-qr}}{2^{-qd}} \label{reducenumden}
	\end{align}
	where we have used that $\widetilde{c_{i,j}} = \mathcal{C}_{i,j} \text{ for all } i,j$ implies $x_1 \widetilde{\rightarrow} y_1, \dots x_q \widetilde{\rightarrow} y_q$ in line (\ref{removex2y}), and that the coins are independent of the round keys to compute the numerator of (\ref{reducenumden}). This completes the lemma.
\end{proof}

We showed earlier that when the round keys are chosen uniformly, they are likely to span $\mathbb{Z}_2^d$. This fact combined with the above lemma means that the tilde process has a near-uniform distribution. So, if we can show that the swap-or-not shuffle has a distribution similar to that of the tilde process, we can show that the swap-or-not shuffle is close to the uniform distribution. To do this we couple the tilde process with the swap-or-not shuffle as follows: \\
\\
Fix $x_1,\dots,x_q$ and $\mathcal{K}_1,\dots,\mathcal{K}_r$. Set $x_i^0 = \widetilde{x_i^0} = x_i$. Then generate $\{\widetilde{c_{i,t}}\}$ as iid Bernoulli$(\frac{1}{2})$ random variables. This defines the tilde process as described above. Now inductively define
$$
c_{i,t} = \begin{cases}
	c_{j,t} \text{ if } x_j^{t-1} + x_i^{t-1} = \mathcal{K}_t \text{ for some } j < i \\
	\widetilde{c_{i,t}} \text{ otherwise}
\end{cases}
$$
The $c_{i,j}$ define the swap-or-not shuffle, as $c_{i,t} = c_{j,t}$ if cards $x_i$ and $x_j$ are paired in round $t$ as required, and otherwise they are independent Bernoulli$(\frac{1}{2})$ random variables.

\begin{definition}
	In the tilde process, we say cards $x_i$ and $x_j$ have a \textbf{collision at time $t$} if $K_t = \widetilde{x_i^{t-1}} + \widetilde{x_j^{t-1}}$ and $(\widetilde{c_{i,t}},\widetilde{c_{j,t}}) \in \{(1,0),(0,1)\}$. In a tilde process with $r$ rounds, we say $x_i$ and $x_j$ have a \textbf{collision} if they have a collision at time $t$ for any $1 \leq t \leq r$.
\end{definition}
That is, $x_i$ and $x_j$ have a collision at time $t$ if $x_i$ ``moves'' to the same position as $x_j$ or vice versa. Note that it is possible to have two cards occupy the same position at time $t$ without a collision at time $t$ if $\widetilde{x_i^{t-1}} = \widetilde{x_j^{t-1}}$ and $\widetilde{c_{i,t}} = \widetilde{c_{j,t}}$. However, if $\widetilde{x_i^{t}} = \widetilde{x_j^{t}}$ then we know that at some time up to and including $t$ the cards $i$ and $j$ collided. \\
\\
Collisions are important because they are the result of a non-injective step and cause the tilde process to ``decouple'' from the swap-or-not shuffle. We can show that in the absence of collisions the swap-or-not shuffle will stay coupled to the tilde process.

\begin{lemma}\label{nocollisionsmeanscoupled}
	Consider the coupled swap-or-not shuffle and tilde process. Fix cards $x_1,\dots,x_q$. Let $M$ be the event that in the tilde process there is at least one collision involving any of these $q$ cards. Then
	$$
	\text{on the event } M^C \text{ we have } c_{i,t} = \widetilde{c_{i,t}} \text{ for all } 1 \leq i \leq q, \ 0 \leq t \leq r
	$$
	where $c_{i,t}$ and $\widetilde{c_{i,t}}$ are the coins used in the swap-or-not shuffle and tilde process respectively.
\end{lemma}
\begin{proof}
	Suppose there exists some $i,t$ such that $c_{i,t} \neq \widetilde{c_{i,t}}$. Then let $i',t'$ be chosen so that $c_{i',t'} \neq \widetilde{c_{i',t'}}$ and so that $t'$ is minimal. Then for all $j \in \{1,\dots,q\}$ and all times $s < t'$ we have $c_{j,s} = \widetilde{c_{j,s}}$. In particular this means that $x_j^{t'-1} = \widetilde{x_j^{t'-1}}$ for all cards $x_j$. \\
	\\
	Since $c_{i',t'} \neq \widetilde{c_{i',t'}}$ there must exist some $j' < i'$ such that $x_{j'}^{t'-1}+x_{i'}^{t'-1} = \mathcal{K}_{t'}$. Since $j'$ is paired with $i'$, and $j'$ is the lesser of the pair, we know that $c_{j',t'} = \widetilde{c_{j',t'}}$. According to our coupling we have $c_{i',t'} = \widetilde{c_{j',t'}}$. Since $c_{i',t'} \neq \widetilde{c_{i',t'}}$ we know
	\begin{equation}
		\widetilde{c_{i',t'}} \neq \widetilde{c_{j',t'}}.
	\end{equation}
	In addition, as $x_{i'}^{t'-1} = \widetilde{x_{i'}^{t'-1}}$ and $x_{j'}^{t'-1} = \widetilde{x_{j'}^{t'-1}}$ we have
	\begin{equation}
		\widetilde{x_{j'}^{t'-1}}+\widetilde{x_{i'}^{t'-1}} = \mathcal{K}_{t'}.
	\end{equation}
	So cards $x_{i'}$ and $x_{j'}$ collide in round $t'$ of the tilde process.
\end{proof}

\begin{corollary}\label{addMCbound} Consider the coupled swap-or-not shuffle and tilde process. Fix cards $x_1,\dots,x_q$. Let $M$ be the event that the tilde process has any pairwise collisions between any of these $q$ cards. Then,
	$$
	\mathbb{P}(x_1 \rightarrow y_1, \dots, x_q \rightarrow y_q) \geq \mathbb{P}(x_1 \widetilde{\rightarrow}y_1, \dots, x_q \widetilde{\rightarrow} y_q, M^C)
	$$
\end{corollary}
\begin{proof}
	As we showed in Lemma \ref{nocollisionsmeanscoupled}, the event $M^C$ implies that $c_{i,t} = \widetilde{c_{i,t}}$ for all $i,t$. So, $M^C$ also implies that for all $i$ we have
	\begin{align}
		\widetilde{x_i^r} &= x_1 + K_1^{\widetilde{c_{i,1}}} + K_2^{\widetilde{c_{i,2}}} + \dots + K_r^{\widetilde{c_{i,r}}} = x_i + K_1^{c_{i,1}} + K_2^{c_{i,2}} + \dots + K_r^{c_{i,r}} = x_i^r.
	\end{align}
	So,
	\begin{align}
		\mathbb{P}(x_1 \rightarrow y_1, \dots, x_q \rightarrow y_q, M^C) = \mathbb{P}(x_1 \widetilde{\rightarrow} y_1, \dots, x_q \widetilde{\rightarrow} y_q, M^C).
	\end{align}
	Hence
	\begin{align}
		\mathbb{P}(x_1 \rightarrow y_1, \dots, x_q \rightarrow y_q) &\geq \mathbb{P}(x_1 \rightarrow y_1, \dots, x_q \rightarrow y_q, M^C) \\
		&= \mathbb{P}(x_1 \widetilde{\rightarrow} y_1, \dots, x_q \widetilde{\rightarrow} y_q, M^C)
	\end{align}
\end{proof}

\subsection{Collisions are Unlikely}

We have shown that if collisions are unlikely then the distribution of the swap-or-not shuffle is close to the distribution of the well-mixed tilde process. This subsection is devoted to showing that collisions are in fact unlikely.

\begin{proposition}\label{collisionsunlikely}
	Consider the tilde process on $N = 2^d$ cards with $r \geq d$ rounds. Fix any $x_i,x_j \in \mathbb{Z}_2^d$. Let $M_{i,j}$ be the event that $x_i$ and $x_j$ have a collision. Then for all $y_i,y_j  \in \mathbb{Z}_2^d$ such that $y_i \neq y_j$ we have,
	$$
	\mathbb{P}(x_i \widetilde{\rightarrow} y_i, x_j \widetilde{\rightarrow} y_j \ | \ M_{i,j}) \leq \frac{7 + 48 \cdot 2^{d-r}}{2(N-1)(N-2)}.
	$$
\end{proposition}
\begin{proof}
	Let $\tau$ be the final time that $x_i$ and $x_j$ collide in the first $r$ rounds. If $x_i$ and $x_j$ do not collide in the first $r$ round, set $\tau = \infty$. A collision between $x_i$ and $x_j$ at time $t$, given the values of $\widetilde{x_i^{t-1}}$ and $\widetilde{x_j^{t-1}}$ happens when $K_t = \widetilde{x_i^{t-1}} + \widetilde{x_j^{t-1}}$ and $\widetilde{c_{i,t}} \neq \widetilde{c_{j,t}}$ with probability $\frac{1}{2} \cdot 2^{-d}$. Note that this probability is the same regardless of the values of $\widetilde{x_i^{t-1}}$ and $\widetilde{x_j^{t-1}}$, a collision at time $t$ is independent of $K_1,\dots,K_{t-1}$ and independent of all coins before time $t$. Let $R_t$ be the filtration recording $K_1,\dots,K_{t}$ and all $c_{k,s}$ with $s \leq t$. Then,
	\begin{equation}
		\mathbb{P}(x_i \text{ and } x_j \text{ collide in round } t \ | \ R_{t-1}) = \frac{1}{2} \cdot 2^{-d}.
	\end{equation}
	independent of $R_{t-1},x_i,x_j$. Thus, if we condition on $\tau = T$ for some $T \leq r$ then the trajectory of $x_i$ and $x_j$ can be described as follows:
	\begin{itemize}
		\item From round $1$ to $T-1$, the round keys and coins for $x_i$ and $x_j$ are chosen uniformly and independently.
		\item In round $T$, the round key is set equal to $\widetilde{x_i^{T-1}} + \widetilde{x_j^{T-1}}$. The coin for $x_i$ in round $T$ is still equally likely to flip heads or tails, but the coin for $x_j$ is fixed to be the opposite. This guarantees $\widetilde{x_i^T} = \widetilde{x_j^T}$.
		\item For a round $s$ between $T+1$ and $r$, the round key and coins are chosen uniformly from all options except $\big(K = \widetilde{x_i^{s-1}} + \widetilde{x_j^{s-1}} \text{ and } (\widetilde{c_{i,s}},\widetilde{c_{j,s}}) \in \{(1,0),(0,1)\}\big)$.
	\end{itemize}
	We can break the possible trajectories into cases. \\
	\\
	Let $B_i$ be the event that $x_i$'s coins flip tails in all rounds strictly \textit{before} $T$, with a similar definition for $B_j$. Let $F_i$ be the event that $x_i$'s coins flip tails in all rounds strictly \textit{after} $T$, with a similar definition for $F_j$. We are concerned with finding upper bounds for the probability $\mathbb{P}(x_i \widetilde{\rightarrow} y_i, x_j \widetilde{\rightarrow} y_j \ | \ M_{i,j})$ for all $y_i \neq y_j$, and we will do this by considering the following cases:
	\begin{enumerate}
		\item $E_1 = F_i \cap F_j$ \\
		In this case, neither $x_i$ nor $x_j$ move from their shared position after $T$, so for all $y_i \neq y_j$,
		\begin{equation}
			\mathbb{P}(x_i \widetilde{\rightarrow} y_i, x_j \widetilde{\rightarrow} y_j \ | \ E_1, \tau = T) = 0.
		\end{equation}
		
		\item $E_2 = B_i^C \cap B_j^C$ \\
		In this case, both $x_i$ and $x_j$ move before $T$. Since the shuffle before the collision uses uniform keys, the locations $x_i$ and $x_j$ move to are uniform (although not necessarily independent, as they may have moved in the same round). Therefore, when they collide at time $T$, their shared position will be uniform, regardless of if $x_i$ or $x_j$ is the one to flip heads. So, due to symmetry, all values of $(y_i,y_j)$ with $y_i \neq y_j$ are equally likely outcomes for $(\widetilde{x_i^r},\widetilde{x_j^r})$. Therefore, for all $y_i \neq y_j$,
		\begin{equation}
			\mathbb{P}(x_i \widetilde{\rightarrow} y_i, x_j \widetilde{\rightarrow} y_j \ | \ E_2, \tau = T) = \frac{\mathbb{P}\left(\widetilde{x_i^r} \neq \widetilde{x_j^r} \ | \ E_2, \tau = T\right)}{N(N-1)} \leq \frac{1}{N(N-1)}
		\end{equation}
		
		\item $E_3 = \big((F_i^C \cap F_j) \ \cup \ (F_i \cap F_j^C)\big) \ \cap \ (B_i \cup B_j)$ \\
		On the event $E_3$ exactly one of the cards exclusively flips tails after $T$. If $\widetilde{x_i^T} = \widetilde{x_j^T} = v$ then on the event $F_i^C \cap F_j$ we have $\widetilde{x_j^r} = v$. By symmetry, all positions other than $v$ are equally likely values for $\widetilde{x_i^r}$. So for all $y_i \neq v$ we have,
		\begin{align}
			&\mathbb{P}\left(x_i \widetilde{\rightarrow} y_i, x_j \widetilde{\rightarrow} v \ | \ \widetilde{x_i^T} = \widetilde{x_j^T} = v, F_i^C \cap F_j, B_i \cup B_j, \tau = T \right) \\
			= \ &\mathbb{P}\left(x_i \widetilde{\rightarrow} y_i \ | \ \widetilde{x_i^T} = \widetilde{x_j^T} = v, F_i^C \cap F_j, B_i \cup B_j, \tau = T \right) \label{removexj} \\
			\leq \ &\frac{\mathbb{P}\left( \widetilde{x_i^r} \neq v \ | \ \widetilde{x_i^T} = \widetilde{x_j^T} = v, F_i^C \cap F_j, B_i \cup B_j, \tau = T\right)}{N-1} \leq \frac{1}{N-1}.
		\end{align}
		Furthermore, note that on the events $x_j \widetilde{\rightarrow} y_j$ and $F_j$ and $\tau = T$, it must be the case that $\widetilde{x_i^T} = \widetilde{x_j^T} = y_j$. Therefore for all $y_i \neq y_j$ we have
		\begin{align}
			&\mathbb{P}\left(x_i \widetilde{\rightarrow} y_i, x_j \widetilde{\rightarrow} y_j \ | \ F_i^C \cap F_j, B_i \cup B_j, \tau = T \right) \\
			\leq \ &\mathbb{P}\left(x_i \widetilde{\rightarrow} y_i, x_j \widetilde{\rightarrow} y_j \ | \ \widetilde{x_i^T} = \widetilde{x_j^T} = y_j, F_i^C \cap F_j, B_i \cup B_j, \tau = T \right) \label{condprob1} \\
			\leq \ &\frac{1}{N-1}
		\end{align}
		where line (\ref{condprob1}) comes from the fact that if $A,Z$ are events with $A \subset Z$ then $\mathbb{P}(A) \leq \mathbb{P}(A \ | \ Z)$. The same argument works for $F_i \cap F_j^C$ so we have
		\begin{equation}
			\mathbb{P}\left(x_i \widetilde{\rightarrow} y_i, x_j \widetilde{\rightarrow} y_j \ | \ F_i^C \cap F_j, B_i \cup B_j, \tau = T \right) \leq \frac{1}{N-1} \text{ for all } y_i \neq y_j.
		\end{equation}
		Therefore, by the union bound,
		\begin{equation}
			\mathbb{P}\left(x_i \widetilde{\rightarrow} y_i, x_j \widetilde{\rightarrow} y_j \ | \ E_3, \tau = T \right) \leq \frac{2}{N-1} \text{ for all } y_i \neq y_j.
		\end{equation}
		
		\item $E_4 = F_i^C \cap F_j^C$ \\
		In this case, both cards flip heads at some point after round $T$. We split this case into three subcases. Let $G$ be the event that $x_i$ and $x_j$ have their first post-$T$ head flip at the same time. Let $H$ be the event that, at the the first time after $T$ a card moves, the round key is the zero vector. Condition on $\widetilde{x_i^T} = \widetilde{x_j^T} = v$. Consider the following subcases:
		\begin{enumerate}
			\item Conditioning on $E_4 \cap G$ \\
			In this subcase, there exists a round $L > T$, where $x_i$ and $x_j$ both flip tails for all rounds between $T$ and $L$, and then both flips heads in round $L$. Since both $x_i$ and $x_j$ flip heads in round $L$, the distribution of the round $L$ key is uniform even after conditioning on $\tau = T$. This effectively puts us in Case 2, as $x_i$ and $x_j$ move together using the round $L$ key to a uniform random position. As in Case 2, due to symmetry, for all $y_i \neq y_j$,
			\begin{equation}
				\mathbb{P}(x_i \widetilde{\rightarrow} y_i,x_j \widetilde{\rightarrow} y_j \ | \ E_4, G, \widetilde{x_i^T} = \widetilde{x_j^T} = v, \tau = T) \leq \frac{1}{N(N-1)}.
			\end{equation}
			\item Conditioning on $E_4 \cap G^C \cap H$ \\
			In this subcase we note that due to symmetry, all targets of the form $(y_i,y_j)$ where $y_i = v$ or $y_j = v$ are equally likely. Similarly, all targets of the form $(y_i,y_j)$ where $y_i \neq y_j$ and $y_i,y_j \neq v$ are equally likely. So for all $y_i \neq y_j$ with $y_i = v$ or $y_j = v$,
			\begin{align}
				&\mathbb{P}(x_i \widetilde{\rightarrow} y_i, x_j \widetilde{\rightarrow} y_j \ | \ E_4, G^C, H, \widetilde{x_i^T} = \widetilde{x_j^T} = v, \tau = T \ ) \\
				= \ &\frac{\mathbb{P}\left(\widetilde{x_i^r} = v \text{ or } \widetilde{x_j^r} = v, \widetilde{x_i^r} \neq \widetilde{x_j^r} \ | \ E_4, G^C, H, \widetilde{x_i^T} = \widetilde{x_j^T} = v, \tau = T \ \right)}{2(N-1)} \\
				\leq \ &\frac{1}{2(N-1)}, \label{higherupperbound}
			\end{align}
			and for all $y_i \neq y_j$ with $y_i,y_j \neq v$,
			\begin{align}
				&\mathbb{P}(x_i \widetilde{\rightarrow} y_i, x_j \widetilde{\rightarrow} y_j \ | \ E_4, G^C, H, \widetilde{x_i^T} = \widetilde{x_j^T} = v, \tau = T \ ) \\
				= \ &\frac{\mathbb{P}\left(\widetilde{x_i^r} \neq v \text{ and } \widetilde{x_j^r} \neq v, \widetilde{x_i^r} \neq \widetilde{x_j^r} \ | \ E_4, G^C, H, \widetilde{x_i^T} = \widetilde{x_j^T} = v, \tau = T \ \right)}{(N-1)(N-2)} \\
				\leq \ &\frac{1}{(N-1)(N-2)}.
			\end{align}
			Since line (\ref{higherupperbound}) provides the higher upper bound, we have for all $y_i \neq y_j$ that,
			\begin{equation}
				\mathbb{P}(x_i \widetilde{\rightarrow} y_i, x_j \widetilde{\rightarrow} y_j \ | \ E_4,G^C,H,\widetilde{x_i^T} = \widetilde{x_j^T} = v,\tau = T \ ) \leq \frac{1}{2(N-1)}.
			\end{equation}
			For future reference, note that
			\begin{equation}
				\mathbb{P}\left( H, G^C \ | \ E_4, \widetilde{x_i^T} = \widetilde{x_j^T} = v, \tau = T \right) \leq \mathbb{P}\left( H \ | \ E_4, \widetilde{x_i^T} = \widetilde{x_j^T} = v, \tau = T \right) = \frac{1}{N}.
			\end{equation}
			\item Conditioning on $E_4 \cap G^C \cap H^C$ \\
			In this subcase, there exists a round $L > T$ where $x_i$ and $x_j$ both flip tails between rounds $T$ and $L$, and in round $L$ either $x_i$ flips heads and $x_j$ flips tails or vice versa. Suppose first that $x_i$ flips heads in round $L$. Then $x_i$ is sent to a uniform position other than $v$. Let $u = \widetilde{x_i^L} \neq v$. It may be the case that $x_i$ flips heads some more times before $x_j$ flips heads, and further moves around, but its position will still be uniform amongst states other than $v$, so without loss of generality assume $x_i$ is still at $u$ in the round before $x_j$ flips heads. When $x_j$ flips heads, all positions other than $u$ are equally likely destinations for $x_j$. Since we are conditioning on no collisions, $x_j$ is half as likely to be sent to $u$ as anywhere else, and if $x_j$ is in fact sent to $u$, we know $x_i$ will swap with $x_j$ and be sent back to $v$. Therefore at time $L$, all positions with $x_i \neq v$ are equally likely, and positions with $x_i = v$ are half as likely. \\
			\\
			If we instead suppose that $x_j$ flips heads in round $L$, then by the same argument all positions with $x_j \neq v$ are equally likely and positions with $x_j = v$ are half as likely. \\
			\\
			Overall this means that, after $x_i$ and $x_j$ have each had their turn to flip heads, all positions with $x_i,x_j \neq v$ are equally likely, and positions with $x_i = v$ or $x_j = v$ are less likely. Now that after the ``flipping heads after $T$'' condition has been met for both $x_i$ and $x_j$, the rest of the shuffle is a standard tilde process except for continuing to condition on no collisions. For the rest of the shuffle, our only bias in round keys is against those that pair $x_i$ and $x_j$'s positions, and force $x_i$ and $x_j$ to swap when they are paired. However, we already have symmetry in probability between states $(a,b)$ and $(b,a)$, regardless of if $a$ or $b$ equals $v$. Therefore, just as in the standard tilde process with collisions allowed, starting with states of the form $(v,b)$ and $(a,v)$ less likely means these states will still be less likely after the final round $r$. So for all $y_i \neq y_j$,
			\begin{align}
				&\mathbb{P}\left( x_i \widetilde{\rightarrow} y_i, x_j \widetilde{\rightarrow} y_j \ | \ E_4, G^C, H^C, \widetilde{x_i^T} = \widetilde{x_j^T} = v, \tau = T \right) \\
				\leq \ &\frac{\mathbb{P}\left(v \neq \widetilde{x_i^r} \neq \widetilde{x_j^r} \neq v \ | \ E_4, G^C, H^C, \widetilde{x_i^T} = \widetilde{x_j^T} = v, \tau = T\right)}{(N-1)(N-2)} \\
				\leq \ &\frac{1}{(N-1)(N-2)}
			\end{align}
		\end{enumerate}
		Now we can combine the three subcases.
		\begin{align}
			&\mathbb{P}\left( x_i \widetilde{\rightarrow} y_i, x_j \widetilde{\rightarrow} y_j \ | \ E_4, \widetilde{x_i^T} = \widetilde{x_j^T} = v, \tau = T \right) \\
			= \ &\mathbb{P}\left( x_i \widetilde{\rightarrow} y_i, x_j \widetilde{\rightarrow} y_j \ | \ E_4, F, \widetilde{x_i^T} = \widetilde{x_j^T} = v, \tau = T \right) \cdot \mathbb{P}\left(F \ | \ E_4, \widetilde{x_i^T} = \widetilde{x_j^T} = v, \tau = T \right) \nonumber \\
			+ &\mathbb{P}\left( x_i \widetilde{\rightarrow} y_i, x_j \widetilde{\rightarrow} y_j \ | \ E_4, G^C, H, \widetilde{x_i^T} = \widetilde{x_j^T} = v, \tau = T \right) \cdot \mathbb{P}\left(G^C,H \ | \ E_4, \widetilde{x_i^T} = \widetilde{x_j^T} = v, \tau = T \right) \nonumber \\
			+ &\mathbb{P}\left( x_i \widetilde{\rightarrow} y_i, x_j \widetilde{\rightarrow} y_j \ | \ E_4, G^C, H^C, \widetilde{x_i^T} = \widetilde{x_j^T} = v, \tau = T \right) \cdot \mathbb{P}\left(G^C,H^C \ | \ E_4, \widetilde{x_i^T} = \widetilde{x_j^T} = v, \tau = T \right) \\
			\leq & \frac{1}{N(N-1)} \cdot 1 + \frac{1}{2(N-1)} \cdot \frac{1}{N} + \frac{1}{(N-1)(N-2)} \cdot 1 < \frac{5}{2(N-1)(N-2)}
		\end{align}
		Since this bound does not depend on $v$, we have
		\begin{equation}
			\mathbb{P}\left( x_i \widetilde{\rightarrow} y_i, x_j \widetilde{\rightarrow} y_j \ | \ E_4, \tau = T \right) \leq \frac{5}{2(N-1)(N-2)}
		\end{equation}
	\end{enumerate}
	Now it is time to combine our four cases. The bounds in cases $E_2$ and $E_4$ are already sufficiently small, but to make the bound in $E_3$ useful we need to incorporate the fact that $E_3$ is unlikely. First note that 
	\begin{align}
		E_3 = \phantom{\cup} \ (B_i \cap B_j \cap F_i^C \cap F_j) &\cup (B_i \cap B_j \cap F_i \cap F_j^C) \nonumber \\
		\cup \ (B_i^C \cap B_j \cap F_i^C \cap F_j) &\cup (B_i^C \cap B_j \cap F_i \cap F_j^C) \nonumber \\
		\cup \ (B_i \cap B_j^C \cap F_i^C \cap F_j) &\cup (B_i \cap B_j^C \cap F_i \cap F_j^C)
	\end{align}
	In other words $E_3$ is given by the union of 6 events, encompassing the outcomes where exactly one card flips all tails after $T$, and at least one card flips all tails before $T$. Note that the probability of any particular card flipping all Tails before $T$ is $2^{-(T-1)}$. The probability of any particular card flipping all Tails after $T$ is $2^{-(r-T)}$. (Heads and Tails are still equally likely, as there is still symmetry after excluding Heads, Tails and Tails, Heads flips with pairing round keys to avoid collision.) Therefore, each of these 6 events has probability bounded above by  $2^{-(T-1)} \cdot 2^{-(r-T)} = 2^{-(r-1)}$. Since there are $6$ events, taking the union bound gives us
	\begin{equation}
		\mathbb{P}(E_3) \leq 6 \cdot 2^{-(r-1)} = 12 \cdot 2^{-r} = \frac{12 \cdot 2^{d-r}}{N}
	\end{equation}
	Now we compute,
	\begin{align}
		\mathbb{P}(x_i \widetilde{\rightarrow} y_i, x_j \widetilde{\rightarrow} y_j \ | \ \tau = T) = &\mathbb{P}(x_i \widetilde{\rightarrow} y_i, x_j \widetilde{\rightarrow} y_j, E_1 \cup E_2 \cup E_3 \cup E_4 \ | \ \tau = T) \\
		\leq \ &\mathbb{P}(x_i \widetilde{\rightarrow} y_i, x_j \widetilde{\rightarrow} y_j \ | \ E_1, \tau = T) \cdot \mathbb{P}(E_1 \ | \ \tau = T) \\
		& \ + \mathbb{P}(x_i \widetilde{\rightarrow} y_i, x_j \widetilde{\rightarrow} y_j \ | \ E_2, \tau = T) \cdot \mathbb{P}(E_2 \ | \ \tau = T) \nonumber \\
		& \ + \mathbb{P}(x_i \widetilde{\rightarrow} y_i, x_j \widetilde{\rightarrow} y_j \ | \ E_3, \tau = T) \cdot \mathbb{P}(E_3 \ | \ \tau = T) \nonumber \\
		& \ + \mathbb{P}(x_i \widetilde{\rightarrow} y_i, x_j \widetilde{\rightarrow} y_j \ | \ E_4, \tau = T) \cdot \mathbb{P}(E_4 \ | \ \tau = T) \nonumber \\
		\leq \ & 0 + \frac{1}{N(N-1)} \cdot 1 + \frac{2}{N-1} \cdot \frac{12 \cdot 2^{d-r}}{N} + \frac{5}{2(N-1)(N-2)} \cdot 1 \\
		\leq \ & \frac{7 + 48 \cdot 2^{d-r}}{2(N-1)(N-2)}
	\end{align}
	Note that this bound is the same regardless of the value of $T \in \{1,\dots,r\}$. Since $\tau \leq r$ is equivalent to $M_{i,j}$ we have,
	\begin{equation}
		\mathbb{P}(x_i \widetilde{\rightarrow} y_i, x_j \widetilde{\rightarrow} y_j \ | \ M_{i,j}) \leq \frac{7 + 48 \cdot 2^{d-r}}{2(N-1)(N-2)}
	\end{equation}
	which is the statement of the theorem.
\end{proof}

\subsection{Uniformity of the Swap-or-Not Shuffle}
We are now ready to prove that the swap-or-not shuffle has a distribution that is close to uniform. We begin by defining a new construction of the tilde process.

\begin{proposition}\label{tildehatequal}
	The tilde process can be defined as follows: \\
	Fix a subset of $q$ distinct cards $x_1,\dots,x_q \in \mathbb{Z}_2^d$. As before, generate uniform independent $K_1,\dots,K_r \in \mathbb{Z}_2^d$. Additionally, generate a uniform $W \in (\mathbb{Z}_2^d)^q$. Now, for any $1 \leq \ell \leq r$, we let
	\begin{align}
		\widetilde{x_1^\ell} &= x_1 + \widehat{c_{1,1}}K_1 + \widehat{c_{1,2}}K_2 + \dots + \widehat{c_{1,k}}K_\ell \\
		\widetilde{x_2^\ell} &= x_2 + \widehat{c_{2,1}}K_1 + \widehat{c_{2,2}}K_2 + \dots + \widehat{c_{2,r}}K_\ell \\
		& \hspace{2cm} \vdots \\
		\widetilde{x_{q\phantom{'}}^\ell} &= x_q + \widehat{c_{q,1}}K_1 + \widehat{c_{q,2}}K_2 + \dots + \widehat{c_{q,r}}K_\ell
	\end{align}
	where $\widehat{c_{i,j}}$ are random elements of $\{0,1\}$ defined as follows:
	\begin{itemize}
		\item If $K_1,\dots,K_r$ span $\mathbb{Z}_2^d$, then, conditioned on the values of $K_1,\dots,K_r$, the coins $\widehat{c_{i,1}},\widehat{c_{i,2}},\dots,\widehat{c_{i,r}}$ are chosen uniformly from all choices such that $\widehat{x_i^r} = W_i$ for all $i$, and independently of all $\widehat{c_{m,j}}$ where $m \neq i$.
		\item If $K_1,\dots,K_r$ do not span $\mathbb{Z}_2^d$, then $\widehat{c_{i,j}}$ are all chosen independently and uniformly from $\{0,1\}$.
	\end{itemize}
\end{proposition}
\begin{proof} In the case that $K_1,\dots,K_r$ do not span $\mathbb{Z}_2^d$ we have by definition that all $\widehat{c_{i,j}}$ are independent, which is consistent with the tilde process. In the case that $K_1,\dots,K_r$ do span $\mathbb{Z}_2^d$ then the disribution of the coins matches that of statement (3) in Lemma \ref{tildeuniform}, so the distribution of the coins is consistent with the tilde process.
\end{proof}

Now that we have shown this new construction for the tilde process, we will from now on assume that the tilde process is generated using $W$.

\begin{lemma}\label{Wtildebound}
	In the tilde process with $r$ rounds, generated using $W$ we have,
	$$
	\mathbb{P}(W_i = y_i, W_j = y_j \ | \ M_{i,j}) \leq \frac{9 + 48 \cdot 2^{-r+d}}{2(N-1)(N-2)}.
	$$
\end{lemma} 
\begin{proof}
	This inequality is similar to the one in Proposition \ref{tildehatequal}. To relate the two inequalities, we must condition on $A_r$, as $A_r$  determines if $W$ fixes the final positions of the cards, or if $W$ is ignored completely. We can decompose $\mathbb{P}(W_i = y_i, W_j = y_j \ | \ M_{i,j})$ as
	\begin{align}
		\mathbb{P}(W_i = y_i, W_j = y_j \ | \ M_{i,j}) &= \mathbb{P}(W_i = y_i, W_j = y_j, A_r \ | \ M_{i,j}) + \mathbb{P}(W_i = y_i, W_j = y_j, A_r^C \ | \ M_{i,j}) \\
		&= \mathbb{P}(x_i \widetilde{\rightarrow} y_i, x_j \widetilde{\rightarrow} y_j, A_r \ | \ M_{i,j}) + \mathbb{P}(W_i = y_i, W_j = y_j, A_r^C \ | \ M_{i,j}) \\
		&\leq \mathbb{P}(x_i \widetilde{\rightarrow} y_i, x_j \widetilde{\rightarrow} y_j \ | \ M_{i,j}) + \mathbb{P}(W_i = y_i, W_j = y_j) \cdot \mathbb{P}(A_r^C \ | \ M_{i,j}) \label{pulloutarC}.
	\end{align}
	In line (\ref{pulloutarC}) we used $\mathbb{P}(W_i = y_i, W_j = y_j \ | \ A_r^C, M_{i,j}) = \mathbb{P}(W_i = y_i, W_j = y_j)$, which is true because on the event $A_r^C$, the value of $W$ is independent of the trajectories of the cards. Since $W$ is uniform, we have $\mathbb{P}(W_i = y_i, W_j = y_j) = \frac{1}{N^2}$. Also recall that Proposition \ref{collisionsunlikely} gave us 
	$$
	\mathbb{P}(x_i \widetilde{\rightarrow} y_i, x_j \widetilde{\rightarrow} y_j \ | \ M_{i,j}) \leq \frac{7+ 48 \cdot 2^{-r+d}}{2(N-1)(N-2)}.
	$$
	Putting this together with (\ref{pulloutarC}), we get,
	\begin{equation}
		\mathbb{P}(W_i = y_i, W_j = y_j \ | \ M_{i,j}) \leq \frac{7+ 48 \cdot 2^{-r+d}}{2(N-1)(N-2)} + \frac{1}{N^2} \cdot 1 \leq \frac{9 + 48 \cdot 2^{-r+d}}{2(N-1)(N-2)}
	\end{equation}
	which completes the lemma.
\end{proof}

\begin{lemma}\label{Wbound}
	Consider the tilde process with $r$ rounds, generated using $W = (w_1,\dots,w_q)$. Let $M$ be the event that there are no pairwise collisions between any of $x_1,\dots,x_q$. Then,
	$$
	\mathbb{P}(W = (y_1,\dots,y_q), M) < \frac{rq(q-1)(9 + 48 \cdot 2^{-r+d})}{4(N-2)N^q}.
	$$
\end{lemma}
\begin{proof}
	To start, we will use the union bound to break up $M$ into it's specific collisions:
	\begin{align}
		\mathbb{P}(W = (y_1,\dots,y_q), M) &= \mathbb{P} \left(\bigcup\limits_{1\leq i < j \leq q} \big\{ W = (y_1,\dots,y_q)
		, M_{i,j} \big\} \right) \\
		&\leq \sum\limits_{1\leq i < j \leq q} \mathbb{P}(W = (y_1,\dots,y_q)
		, M_{i,j})
	\end{align}
	We break the terms in the sum into
	\begin{equation}
		\mathbb{P}(W = (y_1,\dots,y_q), M_{i,j}) = \mathbb{P}(M_{i,j} \ | \ W = (y_1,\dots,y_q)
		) \cdot \mathbb{P}(W = (y_1,\dots,y_q)). \label{wmij1}
	\end{equation}
	Note that $M_{i.j}$ depends only on the trajectories of $x_i$ and $x_j$, and is independent of other cards. So,
	\begin{equation}
		\mathbb{P}(M_{i,j} \ | \ W = (y_1,\dots,y_q)) = \mathbb{P}(M_{i,j} \ | \ W_i = y_i, W_j = y_j) \label{wmij2}.
	\end{equation}
	To compute $\mathbb{P}(M_{i,j} \ | \ W_i = y_i, W_j = y_j)$ we use Bayes' formula:
	\begin{equation}
		\mathbb{P}(M_{i,j} \ | \ W_i = y_i, W_j = y_j) = \frac{\mathbb{P}(M_{i,j})}{\mathbb{P}(W_i = y_i, W_j = y_j)} \cdot \mathbb{P}(W_i = y_i, W_j = y_j \ | \ M_{i,j}). \label{bayes}
	\end{equation}
	Now we need to bound the three probabilities on the RHS of (\ref{bayes}). Since $W$ is uniform, 
	\begin{equation}
		\mathbb{P}(W_i = y_i, W_j = y_j) = \frac{1}{N^2}.
	\end{equation}
	In round $t$ of the shuffle, there is a collision if $K_t = \widetilde{x_i^{t-1}} + \widetilde{x_i^{t-1}}$ and $\widetilde{c_{i,t}} \neq \widetilde{c_{j,t}}$. The round keys and coins are chosen independently and uniformly. There is a $\frac{1}{N}$ chance the round key is chosen to pair $x_i$ and $x_j$, and a $\frac{1}{2}$ chance afterwards that the coins cause a collision. Therefore, the probability of collision in round $t$ is $\frac{1}{2N}$ for all $t$. Using the union bound, we see that the probability of having at least one collision across all the rounds has
	\begin{equation}
		\mathbb{P}(M_{i,j}) \leq \frac{r}{2N}.
	\end{equation}
	Finally, we use the bound for $\mathbb{P}(W_i = y_i, W_j = y_j \ | \ M_{i,j})$, we calculated in Lemma \ref{Wtildebound}:
	\begin{equation}
		\mathbb{P}(W_i = y_i, W_j = y_j \ | \ M_{i,j}) \leq \frac{9 + 48 \cdot 2^{-r+d}}{2(N-1)(N-2)}.
	\end{equation}
	Together, we get
	\begin{equation}
		\mathbb{P}(M_{i,j} \ | \ W_i = y_i, W_j = y_j) \leq \frac{r}{2N} \cdot N^2 \cdot \frac{9 + 48 \cdot 2^{-r+d}}{2(N-1)(N-2)} \leq \frac{r(9 + 48 \cdot 2^{-r+d})}{2(N-2)}
	\end{equation}
	where in the second inequality we used that $\frac{N}{N-1} \leq 2$ as $N \geq 2$. Now we combine this with lines (\ref{wmij1}) and (\ref{wmij2}) to get
	\begin{align}
		\mathbb{P}(W = (y_1,\dots,y_q), M_{i,j}) &= \mathbb{P}(M_{i,j} \ | \ W = (y_1,\dots,y_q)) \cdot \mathbb{P}(W = (y_1,\dots,y_q)) \\
		&= \mathbb{P}(M_{i,j} \ | \ W_i = y_i, W_j = y_j) \cdot \mathbb{P}(W = (y_1,\dots,y_q))\\
		&\leq \frac{r(9 + 48 \cdot 2^{-r+d})}{2(N-2)} \cdot \frac{1}{N^q},
	\end{align}
	where we have used that $\mathbb{P}(W = (y_1,\dots,y_q)) = \frac{1}{N^q}$ due to uniformity. Finally, we sum over all $i \neq j$:
	\begin{equation}
		\mathbb{P}(W = (y_1,\dots,y_q), M) \leq \sum\limits_{1\leq i < j \leq q} \frac{r(9 + 48 \cdot 2^{-r+d})}{2(N-1)} \cdot \frac{1}{N^q} \leq \frac{rq(q-1)(9 + 48 \cdot 2^{-r+d})}{4(N-2)N^q}.
	\end{equation}
\end{proof}

\begin{theorem}\label{swapornotbound}
	Fix $d \in \mathbb{N}$, and $r \geq d$. Fix $x_1,\dots,x_q,y_1,\dots,y_q \in \mathbb{Z}_2^d$. Then, in a swap-or-not shuffle with $r$ rounds and $N = 2^d$ cards,
	$$
	\mathbb{P}(x_1 \rightarrow y_1, \dots, x_q \rightarrow y_q) \geq \frac{1}{(N)_q} \cdot \left( 1 - \frac{q^2}{N} - 2^{-r+d} - \frac{rq(q-1)(9 + 48 \cdot 2^{-r+d})}{4(N-2)} \right)
	$$
\end{theorem} 
\begin{proof}
	We begin considering the coupled tilde process, generated with $W$, and applying Corollary \ref{addMCbound}:
	\begin{align}
		\mathbb{P}(x_1 \rightarrow y_1, x_2 \rightarrow y_2, \dots, x_q \rightarrow y_q) &\geq \mathbb{P}(x_1 \widetilde{\rightarrow} y_1, x_2 \widetilde{\rightarrow} y_2, \dots, x_q \widetilde{\rightarrow} y_q, M^C) \\
		&\geq \mathbb{P}(x_1 \widetilde{\rightarrow} y_1, x_2 \widetilde{\rightarrow} y_2, \dots, x_q \widetilde{\rightarrow} y_q, M^C, A_r) \\
		&= \mathbb{P}(W = (y_1,\dots,y_q), M^C, A_r)
	\end{align}
	and
	\begin{align}
		&\mathbb{P}(W = (y_1,\dots,y_q),M^C,A_r) \\
		\geq \ &\mathbb{P}(W = (y_1,\dots,y_q)) - \mathbb{P}(W = (y_1,\dots,y_q),M) - \mathbb{P}(W = (y_1,\dots,y_q),A_r^C). \label{3terms}
	\end{align}
	
	We now need to bound the three probabilities in (\ref{3terms}). Since $W$ is uniform, we have $\mathbb{P}(W = (y_1,\dots,y_q)) = \frac{1}{N^q}$. We know from Lemma \ref{Arbound} that $\mathbb{P}(A_r^C) = 2^{-r+d}$. Since $W$ is independent of the round keys, we have 
	\begin{equation}
		\mathbb{P}(W = (y_1,\dots,y_q),A_r^C) = \frac{1}{N^q} \cdot 2^{-r+d}
	\end{equation}
	Combining this with our bound for $\mathbb{P}(W = (y_1,\dots,y_q),M)$ from Lemma \ref{Wbound}, we get
	\begin{align}
		\mathbb{P}(W = (y_1,\dots,y_q),M^C \cap A_r) &\geq \frac{1}{N^q} - \frac{1}{N^q} \cdot 2^{-r+d} - \frac{rq(q-1)(9 + 48 \cdot 2^{-r+d})}{4(N-2)N^q} \\
		&= \frac{1}{N^q} \cdot \left( 1 - 2^{-r+d} - \frac{rq(q-1)(9 + 48 \cdot 2^{-r+d})}{4(N-2)} \right)
	\end{align}
	To show small separation distance, our goal is to prove that $\mathbb{P}(x_1 \rightarrow y_1, x_2 \rightarrow y_2, \dots, x_q \rightarrow y_q) \geq (1-\epsilon)\frac{1}{(N)_q}$ for a small $\epsilon$, so it remains to show that $\frac{1}{(N)_q}$ approximately equals $\frac{1}{N^q}$ for sufficiently small $q$. Note that
	\begin{equation}
		\frac{1}{(N)_q} = \frac{1}{N(N-1)\dots(N-q+1)} \leq \frac{1}{(N-q)^q}. \label{fallingfactorial}
	\end{equation}
	Note that for any $a > 1,b \in \mathbb{N}$,
	\begin{equation}
		(a-1)^b = a^b \left( 1-\frac{1}{a} \right) ^b \geq a^b \left(1 - \frac{b}{a} \right),
	\end{equation}
	hence
	\begin{align}
		\frac{1}{(N-q)^q} &= \frac{1}{(q(\frac{N}{q}-1))^q} \\
		&= \frac{1}{q^q(\frac{N}{q}-1)^q} \\
		&\leq \frac{1}{q^q(\frac{N}{q})^q(1-\frac{q^2}{N})} \\
		&= \frac{N^{-q}}{1-\frac{q^2}{N}}.
	\end{align}
	Combining with (\ref{fallingfactorial}) gives
	\begin{align}
		\frac{1}{(N)_q} \cdot \left(1-\frac{q^2}{N}\right) &\leq N^{-q}.
	\end{align}
	Going back to the bound on mixing, we get
	\begin{align}
		\mathbb{P}(x_1 \rightarrow y_1, \dots, x_q \rightarrow y_q) &\geq \frac{1}{(N)_q} \cdot \left(1-\frac{q^2}{N}\right) \cdot \left( 1 - 2^{-r+d} - \frac{rq(q-1)(9 + 48 \cdot 2^{-r+d})}{4(N-2)} \right)\\
		&\geq \frac{1}{(N)_q} \cdot \left( 1 - \frac{q^2}{N} - 2^{-r+d} - \frac{rq(q-1)(9 + 48 \cdot 2^{-r+d})}{4(N-2)} \right)
	\end{align}
\end{proof}

\section{The Security of the Swap-or-Not Shuffle}
In Section 2 we showed that small separation distance leads to good CCA security. In Section 3 we showed that the Swap-or-Not shuffle has small separation distance, provided that the number of queries is not too high. In this section, we will combine these two results to see that the Swap-or-Not shuffle has good CCA security as long as the number of queries is a bit lower than the square root of the number of cards.

\begin{theorem}
	Fix any $\epsilon \in (0,1)$ and $d \geq 2$. Let $X$ be the swap-or-not shuffle with $N = 2^d$ cards, and $r \geq d - \log_2(\epsilon)$ rounds. Consider a CCA adversary equipped with $q \leq \sqrt{\epsilon \cdot \frac{N-2}{r}}$ queries up against this swap-or-not shuffle. The security of $X$ against this adversary is bounded by
	\begin{equation}
		\text{CCA}_q(X) \leq \frac{13}{4}\epsilon + 12\epsilon^2
	\end{equation}
\end{theorem}
\begin{proof}
	Note that $r \geq d - \log_2(\epsilon)$ means $2^{-r+d} \leq \epsilon$, and $q \leq \sqrt{\epsilon \cdot \frac{N-2}{r}}$ means $q^2 \leq \epsilon \cdot \frac{N-2}{r}$. So, plugging $r \geq d - \log_2(\epsilon)$ and $q \leq \sqrt{\epsilon \cdot \frac{N-2}{r}}$ into the bound from Theorem \ref{swapornotbound} gives, for any distinct $x_1,\dots,x_q \in \mathbb{Z}_2^d$ and distinct $y_1,\dots,y_q \in \mathbb{Z}_2^d$ that
	\begin{align}
		\mathbb{P}(x_1 \rightarrow y_1, \dots, x_q \rightarrow y_q) &\geq \frac{1}{(N)_q} \cdot \left( 1 - \frac{q^2}{N} - 2^{-r+d} - \frac{rq(q-1)(9 + 48 \cdot 2^{-r+d})}{4(N-2)} \right) \\
		&\geq \frac{1}{(N)_q} \cdot \left( 1 - \frac{q^2}{N-2} - 2^{-r+d} - \frac{rq^2(9 + 48 \cdot 2^{-r+d})}{4(N-2)} \right)\\
		&\geq \frac{1}{(N)_q} \cdot \left( 1 - \frac{\epsilon (N-2)}{r(N-2)} - \epsilon - \frac{r \epsilon (N-2)(9 + 48 \cdot \epsilon)}{4r(N-2)} \right) \\
		&\geq \frac{1}{(N)_q} \cdot \left( 1 - \frac{13}{4}\epsilon - 12\epsilon^2 \right).
	\end{align}
	Note that under a uniform random permutation, the probability of $(x_1,\dots,x_q)$ being sent to $(y_1,\dots,y_q)$ is $\frac{1}{(N)_q}$. So,
	\begin{equation}
		d_{\text{sep}}(X,\mu) \leq \frac{13}{4}\epsilon + 12\epsilon^2
	\end{equation}
	where $\mu$ is the uniform random permutation. Since this holds for all distinct choices of $q$ queries, we have, by Theorem \ref{ccasepineq},
	\begin{equation}
		\text{CCA}_q(X) \leq \frac{13}{4}\epsilon + 12\epsilon^2
	\end{equation}
\end{proof}
This shows that about $\log_2(N)$ rounds is sufficient for the swap-or-not shuffle on $N$ cards to achieve strong CCA security against an adversary with fewer than $\sqrt{N}$ queries. This lower bound on the number of rounds is tight. To be specific, suppose $Y$ is the swap-or-not shuffle on $N = 2^d$ cards with $d-1$ rounds. Then as long as an adversary has $q > d + \varepsilon$ queries the nCPA security (and therefore the CCA security) of $Y$ is very weak. This is because with $d-1$ rounds the round keys will not span $\mathbb{Z}_2^d$. This means that for each queried card $x_1,\dots,x_q$ the adversary will notice $Y(x_1)-x_1, \dots, Y(x_q)-x_q$ are all in the same subspace. This behavior is unlikely under the uniform random permutation when $q > d$ so $Y$ will have high total variation distance from uniform.

\newpage

\printbibliography

\end{document}